\documentclass[a4paper,11pt]{article}

\usepackage[margin=3cm]{geometry}

\usepackage{microtype}  
\usepackage{amsmath}
\usepackage{amsthm}
\usepackage{hyperref}
\usepackage{tikz}
\usepackage{xspace}

\newcommand{\alleven}{\textup{All\-Cycles\-Even}\xspace}

\newcommand{\limsupodd}{\textup{Limsup\-Odd}\xspace}

\newtheorem{theorem}{Theorem}

\newtheorem{lemma}[theorem]{Lemma}
\newtheorem{claim}[theorem]{Claim}
\newtheorem{corollary}[theorem]{Corollary}

\title{Succinct progress measures for solving parity games\thanks{This
    research has been supported by the EPSRC grant EP/P020992/1
    (Solving Parity Games in Theory and Practice).}} 

\author{Marcin Jurdzi\'nski \qquad Ranko Lazi\'c \\[1ex]
  DIMAP, Department of Computer Science \\
  University of Warwick, UK}

\begin{document}

\date{}

\maketitle

\begin{abstract}
  The recent breakthrough paper by Calude et al.\ has given the first 
  algorithm for solving parity games in quasi-polynomial time, where
  previously the best algorithms were mildly subexponential. 
  We devise an alternative quasi-polynomial time algorithm based on
  progress measures, which allows us to reduce the space required
  from quasi-polynomial to quasi-linear.
  Our key technical tools are a novel concept of ordered tree coding, 
  and a succinct tree coding result that we prove using bounded
  adaptive multi-counters, both of which are interesting in their own
  right.
\end{abstract}

\section{Introduction}

\subsection{Parity games}

A \emph{parity game} is a deceptively simple combinatorial game played 
by two players---Even and Odd---on a directed graph.
From the starting vertex, the players keep moving a token along edges 
of the graph until a lasso-shaped path is formed, that is the first
time the token revisits some vertex, thus forming a loop.   
The set of vertices is partitioned into those owned by Even and those
owned by Odd, and the token is always moved by the owner of the vertex
it is on.
Every vertex is labelled by a positive integer, typically called its 
\emph{priority}. 
What are the two players trying to achieve?
This is the crux of the definition: 
they compete for the highest priority that occurs 
on the loop of the lasso;
if it is even then Even wins, and if it is odd then Odd wins. 

A number of variants of the algorithmic problem of
\emph{solving parity games} are considered in the literature.
The input always includes a game graph as described above. 
The \emph{deciding the winner} variant has an additional part of the
input---the starting vertex---and the question to answer is whether or
not Even has a \emph{winning strategy}---a recipe for winning no
matter what choices Odd makes.
Alternatively, we may expect that the algorithm returns the set of
starting vertices from which Even has a winning strategy, or that it
returns (a representation of) a winning strategy itself; 
the former is referred to as \emph{finding the winning positions}, and
the latter as \emph{strategy synthesis}.

A fundamental result for parity games is
\emph{positional determinacy}~\cite{EJ91,Mos91}:
each position is either winning for Even or winning for Odd, and 
each player has a positional strategy that is winning for her from
each of her winning positions. 
The former is straightforward because parity games---the way we
defined them here---are finite games, but the latter is non-trivial.  
When playing according to a \emph{positional strategy}, in every
vertex that a player owns, she always follows the same outgoing edge,
no matter where the token has arrived to the vertex from.
The answer to the strategy synthesis problem typically is in the form
of a positional strategy succinctly represented by a set of edges: 
(at least) one edge outgoing from each vertex owned by Even. 

Throughout the paper, we write $V$ and $E$ for the sets of vertices
and edges in a parity game graph and $\pi(v)$ for the
(positive integer) priority of a vertex~$v \in V$.
We also use~$n$ to denote the number of vertices;
$\eta$ to denote the numbers of vertices with an odd priority; 
$m$ for the number of edges; and $d$ for the smallest even number that 
is not smaller than the priority of any vertex.
We say that a cycle is \emph{even} if and only if the highest priority
of a vertex on the cycle is even.
We will write $\lg x$ to denote $\log_2 x$, and $\log x$ whenever the
base of the logarithm is moot. 

Parity games are fundamental in logic and verification because
they capture---in an easy-to-state combinatorial game
form---the intricate expressive power of nes\-ting least and greatest 
fixpoint operators (interpreted over appropriate complete lattices),
which play a central role both in the theory and in the practice of
algorithmic verification.    
In particular, the \emph{modal $\mu$-calculus model checking} problem  
is polynomial-time equivalent to solving parity
games~\cite{EJS93}, but parity games are much more broadly
applicable to a multitude of modal, temporal, and
\emph{fixpoint logics}, and in the theory of
\emph{automata on infinite words and trees}~\cite{Gra07}.   

The problem of solving parity games has been found to be both in NP
and in coNP in the early 1990's~\cite{EJS93}.
Such problems are said to be \emph{well characterised}~\cite{Joh07}
and are considered very unlikely to be NP-complete.
Parity games share the rare complexity-theoretic status of being well 
characterised, but not known to be in~P, with such prominent problems
as \emph{factoring}, \emph{simple stochastic games}, and
\emph{mean-payoff games}~\cite{Joh07}.   
Earlier notable examples include \emph{linear programming} and 
\emph{primality}, which were known to be well characterised for many
years before breakthrough polynomial-time algorithms were developed
for them in the late 1970's and the early 2000's, respectively.  

After decades of algorithmic improvements for the modal mu-calculus
model checking~\cite{EL86,BCJLM97,Sei96} and for solving parity 
games~\cite{Jur00,JPZ08,Sch07,CHL15,MRR16}, a recent breakthrough came
from Calude et al.~\cite{CJKLS17} who gave the first algorithm that
works in quasi-polynomial time, where the best upper bounds known
previously were subexponential of the form
$n^{O(\sqrt{n})}$~\cite{JPZ08,MRR16}.
Remarkably, Calude et al.\ have also established fixed parameter
tractability for the key parameter~$d$---the number of distinct vertex  
priorities.   

\subsection{Progress measures}

Our work is inspired by the \emph{succinct counting} technique of
Calude et al.~\cite{CJKLS17}, but it is otherwise rooted in earlier
work on \emph{rankings} and
\emph{progress measures}~\cite{EJ91,KK95,Var96}, and in particular it
is centered on their uses for algorithmically solving games on finite 
game graphs~\cite{Jur00,PP06,Sch07}.    

What is a progress measure?
Paraphrasing Klarlund's~\cite{Kla91,KK95,Kla94a,Kla94c} ideas, 
Vardi~\cite{Var96} coined the following slogans:
\begin{quote}
  A \emph{progress measure} is a mapping on program states that
  quantifies how close each state is to satisfying a property about
  infinite computations.
  On every program transition the progress measure must change in a
  way ensuring that the computation converges toward the property.
\end{quote}
Klarlund and Kozen~\cite{KK95} point out that:
\begin{quote}
  [existence of progress measures] is not surprising from a
  recursion-theoretic point of view [and it] is in essence expressed
  by the Kleene-Suslin Theorem of descriptive set theory,  
\end{quote}
justifying Vardi's~\cite{Var96} admonishment that:
\begin{quote}
  the goal of research in this area should not be merely to prove
  existence of progress measures, but rather to prove the existence of
  progress measures with some \emph{desirable} properties.
\end{quote}
For example, Klarlund~\cite{Kla91,Kla94a}, as well as Kupferman and 
Vardi~\cite{KV98} considered (appropriate relaxations of) progress
measures on infinite graphs and applied them to complementation and
checking emptiness of automata on infinite words and trees.   
Jurdzi\'nski~\cite{Jur00}, Piterman and Pnueli~\cite{PP06}, and
Schewe~\cite{Sch07} focused instead on optimising the magnitude of  
progress measures for Mostowski's parity conditions~\cite{Mos84} and
for Rabin conditions~\cite{KK95} on finite graphs in order to
improve the complexity of solving games with parity, Rabin, and
Streett winning conditions.

In the case of parity games, this allowed Jurdzi\'nski~\cite{Jur00} to
devise the \emph{lifting algorithm} that works in time $n^{d/2+O(1)}$,
where $n$ is the number of vertices and $d$ is the number of distinct
vertex priorities.  
Schewe~\cite{Sch07} improved the running time to $n^{d/3+O(1)}$ by
combining the divide-and-conquer \emph{dominion} technique of
Jurdzi\'nski et al.~\cite{JPZ08} with a modification of the lifting
algorithm, using the latter to detect medium-sized dominions more
efficiently.  

\subsection{Our contribution}

We follow the work of Jurdzi\'nski~\cite{Jur00} and
Schewe~\cite{Sch07} who have developed efficient algorithms for
solving parity games by proving existence of
\emph{small progress measures}.
Our contribution is to prove that every progress measure on a finite 
game graph is---in an appropriate sense---equivalent to a
\emph{succinctly} represented progress measure.  
This paves the way to the design of an algorithm that slightly
improves the quasi-polynomial time complexity of the algorithm of
Calude et al.~\cite{CJKLS17}, and that significantly improves the
space complexity from quasi-polynomial down to quasi-linear.  

More specifically and technically, we argue that
\emph{navigation paths} from the root to nodes in ordered trees of 
height~$h$ and with at most~$n$ leaves can be succinctly encoded using
at most approximately $\lg h \cdot \lg n$ bits by means of  
\emph{bounded adaptive multi-counters}.
The statement and the proof of this \emph{tree coding} result are
entirely independent of parity games, and they are notable in their
own right.
The concept of ordered tree coding that we introduce seems fundamental
and it may find unrelated applications.  

Our application of the tree coding result to parity games is based on
the fact that a progress measure for a graph with~$n$ vertices and~$d$  
distinct vertex priorities can be viewed as a labelling of vertices
by (the navigation paths from the root to) leaves of an ordered 
tree of height~$d/2$ and with at most~$n$ leaves.
It then follows that there are approximately at most
$2^{\lg d \cdot \lg n} = n^{\lg d}$ 
possible encodings to consider for every vertex,
a considerable gain over the naive bound
$2^{d/2 \cdot \lg n} = n^{d/2}$ that determined the complexity of
Jurdzi\'nski's~\cite{Jur00} algorithm.
We argue, however, that the lifting technique developed by
Jurdzi\'nski~\cite{Jur00} can be adapted to iteratively compute a
succinct representation of a progress measure in the quasi-polynomial
time $O\left(n^{\lg d}\right)$ and quasi-linear space
$O(n \log n \cdot \log d)$.    

\subsection{Related work}

The high-level idea of the algorithm of Calude et al.~\cite{CJKLS17} 
bears similarity to the approach of Bernet et al.~\cite{BJW02}:
first devise a finite safety automaton that recognizes infinite
sequences of priorities that result in a win for Even
(in the case of Bernet et al., given an explicit upper bound on the
number of occurrences of each odd priority before an occurrence of a
higher priority),
and then solve the safety game obtained from the product automaton
that simulates the safety automaton on the game graph. 

The key innovation of Calude et al.\ is their
\emph{succinct counting technique} which allows them to devise a
finite (safety) automaton
(not made explicit, but easy to infer from their work) with only
$n^{O(\log d)}$ states, while that of Bernet et al.\ may have
$\Omega\left((n/d)^{d/2}\right)$ states.
On the other hand, Calude et al.\ construct the safety game explicitly
before solving it, thus requiring not only quasi-polynomial time but
also quasi-polynomial space, and not just in the worst case but
always.  
In contrast, Bernet et al.\ develop a technique for solving the
safety game symbolically without explicitly constructing it, hence
avoiding superpolynomial space complexity;
as they point out: ``The algorithm actually turns out to be the same
as [the lifting] algorithm of Jurdzinski''~\cite{BJW02}
(although, in fact, they bring down rather than lift up).

Contemporaneously and independently from the early version of our
work~\cite{JL17}, Fearnley et al.~\cite{FJSSW17} have developed a
technique of lifting Calude et al.'s~\cite{CJKLS17}
\emph{play summaries} so as to efficiently solve Calude et al.'s
safety game without constructing it explicitly, and Gimbert and
Ibsen-Jensen~\cite{GIJ17} have given slightly improved upper bounds on
the running time of Calude et al.'s algorithm. 
While our succinct progress measures and bounded adaptive
multi-counters are notably different from Calude et al.'s and Fearnley
et al.'s play summaries, the complexity bounds achieved by us, by
Fearnley et al., and by Gimbert and Ibsen-Jensen are remarkably
similar.   
For the benchmark case when $d \leq \lg n$, Calude et al.\ gave the
$O(n^5)$ upper bound on the running time, and our $O(mn^{2.38})$ bound
is slightly better than the $O(mn^{2.55})$ improved bound derived for
Calude et al.'s algorithm by Gimbert and Ibsen-Jensen.
In the general case, Calude et al.\ gave the $O(n^{\lg d + 6})$
upper bound on the running time of their algorithm for finding the
winning positions and $O(n^{\lg d + 7})$ for strategy synthesis. 
For the case $d = \omega(\log n)$, we establish the
$O\left(dm\eta^{\lg(d/{\lg \eta})+1.45}\right)$ running time upper
bound, which is roughly the same as the one obtained by Fearnley et
al., and Gimbert and Ibsen-Jensen achieve the analogous
$O\left(dmn^{\lg(d/{\lg n})+1.45}\right)$ upper bound for Calude et
al.'s algorithm. 
Notably, however, both Fearnley et al.\ and Gimbert and Ibsen-Jesen
claim those bounds only for the cases $d \geq \log^2 \eta$ and 
$d = \Omega(\log^2 n)$, respectively.  

Boja\'nczyk and Czerwi\'nski~\cite[Chapter~3]{BC18} have recently developed
a reworking of the algorithm of Calude et al.~\cite{CJKLS17} that is based on
constructing a deterministic safety automaton of quasi-polynomial size that separates
the language of all infinite words of vertices in which all cycles are even 
from its odd counterpart.
We supplement the main results in this paper by developing separating automata 
of quasi-polynomial size that are based on the bounded adaptive multi-counters.
They seem significantly different from the separating automata of 
Calude et al.\ cum Boja\'nczyk and Czerwi\'nski which are based on play summaries.
Moreover, both the construction and the proof of correctness 
are perhaps surprisingly simple.

\section{Succinct tree coding}
\label{section:tree-coding}

What is an \emph{ordered tree}?
One formalisation is that it is a prefix-closed set of sequences of
elements of a linearly ordered set. 
For clarity, in contrast to graphs, we refer to those sequences as
\emph{nodes}, and the maximal nodes (w.r.t.\ the prefix ordering) are
called \emph{leaves}.   
The \emph{root} of the tree is the empty sequence, sequences of
length~$1$ are the \emph{children} of the root, sequences of
length~$2$ are their children, and so on.
We also refer to the elements of the linearly ordered set that occur
in the sequences as \emph{branching directions}: 
for example, if we use the non-negative integers with the usual
ordering as branching directions, then the node $(3, 0, 5)$ is the
child of the node $(3, 0)$ reached from it via branching
direction~$5$.  
Moreover, we refer to the sequences of branching directions
that uniquely identify nodes as their \emph{navigation paths}.  

What do we mean by \emph{ordered-tree coding}?
The notion we find useful in the context of this work is an
order-preserving relabelling of branching directions, allowing for the
relabellings at various nodes to differ from one another
(or, in other words, to be \emph{adaptive}).
The intention when coding in this way is to obtain an isomorphic
ordered tree, and the intended purpose is to be able to more
succinctly encode the navigation paths for each leaf in the tree.

Succinct codes are easily obtained if trees are well balanced.
As a warm-up, consider the ordered tree of height~$1$ and with~$\ell$
leaves 
(that is, the tree consists of the root whose $\ell$ children are all
leaves).  
Whatever the (identity of the) branching directions from the root to
the~$\ell$ leaves are, for every~$i = 0, 1, \dots, \ell-1$, we can
relabel the branching direction of the $i$-th
(according to the linear order on branching directions) child of the
root to be the binary representation of the number~$i$, which shows
that the navigation path of every leaf in the tree can be described
using only $\lceil \lg \ell \rceil$ bits.   
The reader is invited to verify that increasing height while
maintaining balance of such ordered trees does not (much) increase the
number of bits (expressed as a function of the number of leaves)
needed to encode the navigation paths.
Consider for example the case of a perfect $k$-ary tree of
height~$h$; 
it has $\ell = k^h$ leaves and every navigation path can be encoded by
$h \cdot \lceil \lg k \rceil$ bits, $\lceil \lg k \rceil$ bits
per each $k$-ary branching;
argue that it is bounded by $2 \lg \ell$ for all $k \geq 2$, and is in
fact $(1+o(1)) \cdot \lg \ell$. 

But what if the tree is not nearly so well balanced?
How many more bits may be needed to encode navigation paths in
arbitrary trees of height~$h$ and with~$\ell$ leaves?
The key technical result of this section, on which the main results of 
the paper hinge, is that---thanks to adaptivity of our notion of
ordered-tree coding---$(\lceil \lg h \rceil + 1) \lceil \lg \ell \rceil$
bits suffice. 

We define the set $B_{g, h}$ of
\emph{$g$-bounded adaptive $h$-counters}
to consist of $h$-tuples of binary strings whose total length is at
most~$g$.
For example, $(0, \varepsilon, 1, 0)$ and
$(\varepsilon, 1, \varepsilon, 0)$ are $3$-bounded adaptive
$4$-counters,
but $(0, 1, \varepsilon, \varepsilon, 0)$ and
$(10, \varepsilon, 01, \varepsilon)$ are not---the former is a
$5$-tuple, and the total length of the binary strings in the latter
is~$4$. 

We define a strict linear ordering~$<$ on finite binary strings
as follows, for both binary digits $b$, and for all binary strings~$s$
and~$s'$:
\begin{equation}
  \label{equation:bin-str-order}
  0 s < \varepsilon, \qquad
  \varepsilon < 1 s, \qquad
  b s < b s' \;\;\text{iff}\;\; s < s'.
\end{equation}
Equivalently, it is the ordering on the rationals obtained by the mapping
\[b_1 b_2 \cdots b_k \,\mapsto\,
  \sum_{i = 1}^k (-1)^{b_i + 1} 2^{-i}\;.\]

We extend the ordering to $B_{g, h}$ lexicographically.
For example, $(00, \varepsilon, 1) < (0, 0, 0)$ because $00 < 0$, and
$(\varepsilon, 011, 1) < (\varepsilon, \varepsilon, 000)$ because
$011 < \varepsilon$. 
 
\begin{lemma}[Succinct tree coding]
  \label{lemma:succinct-tree-coding}
  For every ordered tree of height~$h$ and with at most~$\ell$ leaves 
  there is a tree coding in which every navigation path is an 
  $\lceil \lg \ell \rceil$-bounded adaptive $i$-counter, where
  $i \leq h$ is the length of the path. 
\end{lemma}

\begin{proof}
We argue inductively on~$\lceil \lg \ell \rceil$ and~$h$.

The base case, $\ell = 1$ and $h = 0$, is trivial.

Let $M$ be a branching direction from the root 
such that both sets of leaves:  
$L_<$ whose first branching direction (i.e., from the root) 
is strictly smaller than~$M$, and 
$L_>$ whose first branching direction is strictly larger than~$M$, 
are of size at most~$\ell/2$.
Also let $L_=$ to be the set of leaves whose first branching direction is~$M$. 
The required coding is obtained in the following way:
\begin{itemize}
\item
  If $L_< \not= \emptyset$, apply the inductive hypothesis to the
  subtree with $L_<$ as the set of leaves, and append one leading~$0$
  to the binary strings that code the first branching direction. 
\item
  If $L_> \not= \emptyset$, apply the inductive hypothesis to the
  subtree with $L_>$ as the set of leaves, and append one leading~$1$
  to the binary strings that code the first branching direction. 
\item
  Let the empty binary string $\varepsilon$ be the code of the
  branching direction~$M$ from the root of the tree, and then obtain
  the required coding of the rest of the subtree rooted at node $(M)$
  by applying the inductive hypothesis for trees of height at
  most~$h-1$ and with at most $\ell$ leaves.
  \qedhere
\end{itemize}
\end{proof}

The lemma is illustrated, for an ordered tree of height $2$ and with $8$ leaves,
in Figures \ref{figure:small} and~\ref{figure:succinct}.

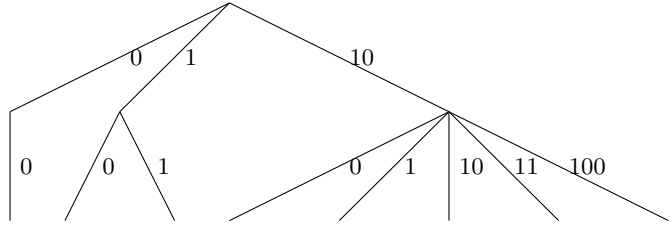
\begin{figure}
\begin{center}
\begin{tikzpicture}
  [ grow                    = down,
    sibling distance        = 3.75em,
    level distance          = 3.75em,
    edge from parent/.style = {draw, right, font=\footnotesize} ]
  \coordinate
  child { child { edge from parent node {$0$} } 
          edge from parent node {$0$} }
  child { child { edge from parent node {$0$} }
          child { edge from parent node {$1$} }
          edge from parent node {$1$}}
  child[missing]
  child[missing]
  child { child { edge from parent node {$0$} }
          child { edge from parent node {$1$} }
          child { edge from parent node {$10$} }
          child { edge from parent node {$11$} }
          child { edge from parent node {$100$} } 
          edge from parent node {$10$} };
\end{tikzpicture}
\end{center}
\caption{The branching directions are simply numbered in binary.
For instance, the navigation path to the right-most leaf is $(10,100)$, 
which uses $5$ bits.}
\label{figure:small}
\end{figure}

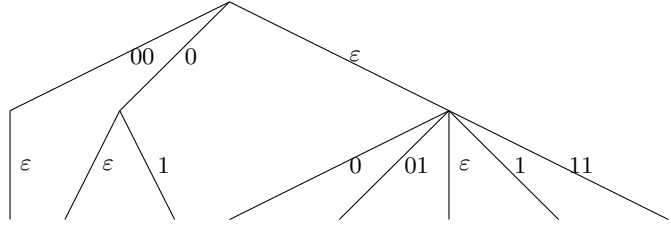
\begin{figure}
\begin{center}
\begin{tikzpicture}
  [ grow                    = down,
    sibling distance        = 3.75em,
    level distance          = 3.75em,
    edge from parent/.style = {draw, right, font=\footnotesize} ]
  \coordinate
  child { child { edge from parent node {$\varepsilon$} } 
          edge from parent node {$00$} }
  child { child { edge from parent node {$\varepsilon$} }
          child { edge from parent node {$1$} }
          edge from parent node {$0$}}
  child[missing]
  child[missing]
  child { child { edge from parent node {$0$} }
          child { edge from parent node {$01$} }
          child { edge from parent node {$\varepsilon$} }
          child { edge from parent node {$1$} }
          child { edge from parent node {$11$} } 
          edge from parent node {$\varepsilon$} };
\end{tikzpicture}
\end{center}
\caption{The same ordered tree is coded succinctly so that, 
in every navigation path, the total number of bits is at most 
$\lceil \lg 8 \rceil = 3$ (in this example, it happens to be at most~$2$).}
\label{figure:succinct}
\end{figure}

Note that the tree coding lemma implies the promised 
$(\lceil \lg h \rceil + 1) \lceil \lg \ell \rceil$ upper bound:
every $\lceil \lg \ell \rceil$-bounded adaptive $h$-counter can be
coded by a sequence of $\lceil \lg \ell \rceil$ single bits, each
followed by the $\lceil \lg h \rceil$-bit representation of the
number of the component that the single bit belongs to 
in the bounded adaptive multi-counter. 
In Section~\ref{section:lifting-algorithm} we give more refined
estimates of the size of the set
\[
S_{\eta, d} =
\bigcup_{i=0}^{d/2} B_{\lceil \lg \eta \rceil, i}
\]
of $\lceil \lg \eta \rceil$-bounded adaptive $i$-counters,
where $0 \leq i \leq d/2$, which is the dominating term in the
worst-case running time bounds of our lifting algorithm for solving
parity games.

\section{Succinct progress measures}
\label{section:spm}

\subsection{Progress measures}

For finite parity game graphs, a progress measure~\cite{Jur00} is a
mapping from the~$n$ vertices to $d/2$-tuples of non-negative integers
(that also satisfies the so-called progressiveness conditions on an 
appropriate set of edges, as detailed below).  
Note that an alternative interpretation is that a progress measure
maps every vertex to a leaf in an ordered tree~$T$
in which each of the at most $n$ leaves has a navigation path of
length~$d/2$. 

What are the conditions that such a mapping needs to satisfy to be a
progress measure?
For every priority $p \in \{1, 2, \dots, d\}$, we obtain the
\emph{$p$-truncation} $(r_{d-1}, r_{d-3}, \dots, r_1)|_p$ of the 
$d/2$-tuple $(r_{d-1}, r_{d-3}, \dots, r_1)$ of non-negative integers,
one per each odd priority, by removing the components corresponding to
all odd priorities~$i$ lower than~$p$. 
For example, we have $(2, 7, 1, 4)|_8 = \varepsilon$,
$(2, 7, 1, 4)|_5 = (2, 7)$ and $(2, 7, 1, 4)|_2 = (2, 7, 1)$. 
We compare tuples using the lexicographic order.
We say that an edge $(v, u) \in E$ is \emph{progressive} in~$\mu$ if 
\[
\mu(v)|_{\pi(v)} \geq \mu(u)|_{\pi(v)}\,,
\]
and the inequality is strict when $\pi(v)$ is odd.
Finally, the mapping $\mu : V \to T$ is a
\emph{progress measure}~\cite{Jur00} if:   
\begin{itemize}
\item
  for every vertex owned by Even, some outgoing edge is progressive
  in~$\mu$, and
\item
  for every vertex owned by Odd, every outgoing edge is progressive
  in~$\mu$. 
\end{itemize}

\subsection{Trimmed progress measures}

Observe that the progressiveness condition of every edge
$(v, u) \in E$ is formulated by referring to the $\pi(v)$-truncations
of the tuples labelling vertices~$v$ and~$u$, so if the label of~$v$
is $(r_{d-1}, r_{d-3}, \dots, r_1)$ then the components $r_i$ for
$i < \pi(v)$ are superfluous for stating the condition.
It is therefore reasonable to consider
\emph{trimmed progress measures} that label vertices with tuples
$(r_{d-1}, r_{d-3}, \dots, r_{k+2}, r_k)$ of length at most $d/2$,
rather than insisting on all vertices having labels of length
exactly~$d/2$.  
In the alternative interpretation discussed above, such a trimmed 
progress measure may then map some vertices to nodes in an ordered
tree (of height at most $d/2$) that are not leaves.

We clarify that if two tuples of different lengths are to be compared
lexicographically, and if the shorter one is a prefix of the longer
one, then the shorter one is defined to be lexicographically strictly
smaller than the longer one
For example, we have $(1, 0) < (1, 0, 3)$, 
but $(1, 0, 3) < (1, 1)$.
Moreover, truncations of tuples of length smaller than $d/2$ are
defined analogously;
in particular, if $p \leq k$ then
$(r_{d-1}, r_{d-3}, \dots, r_k)|_p = (r_{d-1}, r_{d-3}, \dots, r_k)$. 

\subsection{Succinct progress measures}

It is well known that existence of a progress measure is sufficient
and necessary for existence of a winning strategy for Even from every
starting vertex~\cite{EJ91,KK95,Jur00}.
Our main contribution in this section is the observation
(Lemma~\ref{lemma:necessity})
that this is also true for existence of a
\emph{succinct progress measure}, in which the ordered tree~$T$ is
such that:  
\begin{itemize}
\item
  finite binary strings ordered as in~(\ref{equation:bin-str-order})
  are used as branching directions instead of non-negative integers,
  and 
\item
  for every navigation path, the sum of lengths of the binary strings
  used as branching directions is at most
  $\lceil \lg \eta \rceil$;   
\end{itemize}
or in other words, that every navigation path in~$T$ is a 
$\lceil \lg \eta \rceil$-bounded adaptive $i$-counter, for
some $i, 0 \leq i \leq d/2$;
recall that $\eta$ is the number of vertices with an odd priority. 

In succinct progress measures, truncations and lexicographic ordering
of tuples, as well as progressiveness of edges, are defined
analogously.
Again, we clarify that if two tuples of different lengths are to be
compared lexicographically, and if the shorter one is a prefix of the
longer one, then the shorter one is defined to be lexicographically
strictly smaller than the longer one. 
For example,
we have $(01, \varepsilon) < (01, \varepsilon, 00)$, 
but $(01, \varepsilon, 000) < (1000, \varepsilon)$.

\subsection{Sufficiency}

Sufficiency does not require a new argument because the standard 
reasoning---for example as in~\cite[Proposition~4]{Jur00}---relies
only on the ordered tree structure (through truncations), and not on
what ordered set is used for the branching directions.     
We provide a proof here for completeness.

\begin{lemma}[Sufficiency]
  \label{lemma:sufficiency}
  If there is a succinct progress measure then there is a positional
  strategy for Even that is winning for her from every starting
  vertex. 
\end{lemma}

\begin{proof}
  Let $\mu$ be a succinct progress measure.  
  Let Even use a positional strategy that only follows edges that are
  progressive in~$\mu$. 
  Since every edge outgoing from vertices owned by Odd is also
  progressive in~$\mu$, it follows that only progressive edges will be
  used in every play consistent with the strategy.  
  Therefore, in order to verify that the strategy is winning for Even,
  it suffices to prove that if all edges in a simple cycle are
  progressive in~$\mu$ then the cycle is even.

  Let $v_1, v_2, \dots, v_k$ be a simple cycle
  in which all edges $(v_1, v_2)$, $(v_2, v_3)$, \dots,
  $(v_{k-1}, v_k)$, and $(v_k, v_1)$ are progressive in~$\mu$.
  For the sake of contradiction, suppose that the highest priority~$p$ 
  that occurs on the cycle is odd, and without loss of generality, let
  $\pi(v_1) = p$. 
  By progressivity of all the edges on the cycle, we have that  
  \[
  \mu(v_1) {|_p} > \mu(v_2) {|_p} \geq \cdots
  \geq \mu(v_k) {|_p} \geq \mu(v_1) {|_p}\;,
  \]
  absurd.
\end{proof}

\subsection{Necessity}

We prove necessity by first slightly strengthening the existence of
the least progress measure result~\cite[Theorem 11]{Jur00},  
and then by applying the succinct tree coding lemma
(Lemma~\ref{lemma:succinct-tree-coding}).  

As discussed earlier in this section, the range of a progress
measure is the set of nodes in a tree of height~$d/2$ and with
at most~$n$ leaves.  
By applying Lemma~\ref{lemma:succinct-tree-coding} to this tree we may
conclude that there is a tree coding in which branching directions on
every navigation path use at most $\lceil \lg n \rceil$ bits. 
This way we come short [sic] of satisfying our definition of a
succinct progress measure: 
the definition allows us, on every path, to use at most 
$\lceil \lg \eta \rceil$ bits for branching directions,
which may be strictly smaller than $\lceil \lg n \rceil$.

In order to overcome this hurdle, we first define the operation of a 
trimming of a progress measure. 
Let $\mu$ be a progress measure.
We define the \emph{trimming} $\mu^\downarrow$ of~$\mu$ as follows:
for every vertex~$v$, we let $\mu^\downarrow(v)$ be the longest prefix
of $\mu(v)$ whose last component is not~$0$;
in particular, if $\mu(v)$ is a sequence of $0$s of length $d/2$ then 
$\mu^\downarrow(v)$ is the empty sequence. 
For convenience, we also define the inverse operation:
if $\mu$ is a trimmed progress measure then for every vertex~$v$, 
we let $\mu^\uparrow(v)$ be the sequence of length $d/2$ obtained by
adding an appropriate number (possibly none) of $0$s at the end
of~$\mu(v)$.  

Recall that by~\cite[Corollary 8]{Jur00} and by the proof 
of~\cite[Theorem 11]{Jur00} the \emph{least progress measure}
$\mu_*$ exists, where the relevant order on mappings from vertices to
sequences of non-negative integers is pointwise lexicographic. 

\begin{lemma}
  \label{lemma:n-odd-leaves}
  The trimming $\mu_*^\downarrow$ of the least progress
  measure~$\mu_*$ is a trimmed progress measure and the ordered tree
  that it maps to has at most~$\eta$ leaves.  
\end{lemma}

\begin{proof}
  That $\mu_*^\downarrow$ is a trimmed progress measure follows
  routinely from $\mu_*$ being a progress measure and from the
  definion of a trimmed progress measure.

  We argue that for every leaf in the ordered tree $T$ that
  $\mu_*^\downarrow$ maps into, there is a vertex~$v \in V$
  with an odd priority, such that $\mu_*^\downarrow(v)$ is that leaf,
  which implies the other claim of the lemma.

  Let $\lambda = (r_{d-1}, r_{d-3}, \dots, r_k)$ be a leaf in~$T$.
  For the sake of contradiction, assume that every vertex $v \in V$
  such that $\mu_*^\downarrow(v) = \lambda$ has an even priority.
  Note that---by the definition of a trimming---$r_k \not= 0$, and
  hence---because $\mu_*$ is the least progress
  measure---$k > \pi(v)$.
  We define the mapping $\mu_\lambda$ as follows:
  \[\mu_\lambda(v) =
  \begin{cases}
    \mu_*^\downarrow(v) & \text{if } \mu_*^\downarrow(v) \not= \lambda\,, 
    \\
    (r_{d-1}, \dots, r_{k+2}, r_k - 1, n+1) &
      \text{if } \mu_*^\downarrow(v) = \lambda \text{ and } k \geq 3\,,
    \\
    (r_{d-1}, \dots, r_3, r_1 - 1) &
      \text{if } \mu_*^\downarrow(v) = \lambda \text{ and } k=1\,,
  \end{cases}
  \]
  for all $v \in V$.
  We argue that $\mu_\lambda$ is a trimmed progress measure, and hence
  $\mu_\lambda^\uparrow$ is a progress measure, which---because
  $\mu_\lambda^\uparrow$ is strictly smaller than $\mu_*$---duly 
  contradicts the assumption that~$\mu_*$ was the least progress
  measure.  

  We only need to verify that every edge $(v, u) \in E$, such that
  $\mu_*^\downarrow(v) = \lambda$, and that is progressive
  in~$\mu_*^\downarrow$, is also progressive in~$\mu_\lambda$.
  This is straightforward if $\mu_*^\downarrow(u) = \lambda$.
  Otherwise, we have
  $\mu_*^\downarrow(v)|_{\pi(v)} > \mu_*^\downarrow(u)|_{\pi(v)}$,
  which implies that
  \begin{equation}
    \label{equation:geq}
    (r_{d-1}, \dots, r_{k+2}, r_k - 1) \, \geq \,
    \mu_*^\downarrow(u)|_k\,.
  \end{equation}
  We now argue that
  $\mu_\lambda(v)|_{\pi(v)} \geq \mu_\lambda(u)|_{\pi(v)}$ by
  considering the following two cases.
  \begin{itemize}
  \item
    If $\pi(v) = k-1$ then
    \[
    \mu_\lambda(v)|_{\pi(v)} \, = \,
    (r_{d-1}, \dots, r_{k+2}, r_k - 1) \, \geq \,
    \mu_\lambda(u)|_{\pi(v)}\,,
    \]
    where the inequality follows from~(\ref{equation:geq}). 
    
  \item
    If $\pi(v) < k-1$ then
    \[
    \mu_\lambda(v)|_{\pi(v)} \, = \,
    (r_{d-1}, \dots, r_{k+2}, r_k - 1, n+1) \, > \,
      \mu_*^\downarrow(u)|_{\pi(v)} = \mu_\lambda(u)|_{\pi(v)}\,,
    \]
    where the inequality follows from~(\ref{equation:geq}) and because
    no component in the least progress measure~$\mu_*$ exceeds~$n$.
    \qedhere
  \end{itemize}
\end{proof}

Necessity now follows from applying the succinct coding lemma
(Lemma~\ref{lemma:succinct-tree-coding}) to the tree with at most
$\eta$ leaves, which is obtained by
Lemma~\ref{lemma:n-odd-leaves}.

\begin{lemma}[Necessity]
  \label{lemma:necessity}
  If there is a strategy for Even that is winning for her from every
  starting vertex, then there is a succinct progress measure.
\end{lemma}

\section{Lifting algorithm}
\label{section:lifting-algorithm}

Without loss of generality, we assume that $\eta \leq n/2$.
Otherwise, we have that the number of vertices with an even priority 
is less than $n/2$, and we can apply the algorithm below to 
the \emph{dual} game obtained by reducing the priority of each vertex
by~$1$ and exchanging the roles of the two players;   
the winning set and a winning strategy for player Even in the dual
game computed by the algorithm are the winning set and a winning
strategy for player Odd in the original one. 
Note that the algorithm can be applied without any asymptotic penalty
(and only cosmetic changes) to games in which vertices of priority~$0$
are allowed, and hence the analysis of the algorithm applies to both
the original game and its dual.

\subsection{Algorithm design and correctness}

Consider the following linearly ordered set of bounded adaptive 
multi-counters: 
\[
S_{\eta, d} =
\bigcup_{i=0}^{d/2} B_{\lceil \lg \eta \rceil, i}\;, 
\]
and let $S_{\eta, d}^\top$ denote the same set with an
extra top element~$\top$.  
We extend the notion of succinct progress measures 
to mappings $\mu : V \to S_{\eta, d}^\top$ by: 
\begin{itemize}
\item 
defining the truncations of $\top$ as $\top {|_p} = \top$ for all~$p$;
\item
regarding edges $(v, u) \in E$ such that 
$\mu(v) = \mu(u) = \top$ and $\pi(v)$ is odd 
as progressive in~$\mu$.
\end{itemize} 

For any mapping $\mu : V \to S_{\eta, d}^\top$ and edge $(v, w) \in E$, 
let $\mathrm{lift}(\mu, v, w)$ be the least $\sigma \in S_{\eta, d}^\top$
such that $\sigma \geq \mu(v)$ and $(v, w)$ is progressive in
$\mu[v \mapsto \sigma]$. 
For any vertex $v$, we define an operator $\mathrm{Lift}_v$ 
on mappings $V \to S_{\eta, d}^\top$ as follows:
\[
\mathrm{Lift}_v(\mu)(u) =
\begin{cases}
\mu(u)                                       & \text{if } u \neq v, \\
\min_{(v, w) \in E} \mathrm{lift}(\mu, v, w) & \text{if Even owns } u = v, \\
\max_{(v, w) \in E} \mathrm{lift}(\mu, v, w) & \text{if Odd owns } u = v.
\end{cases}
\]

\begin{theorem}[Correctness of lifting algorithm]
  \mbox{} 
\begin{enumerate}
\item 
The set of all mappings $V \to S_{\eta, d}^\top$ ordered pointwise 
is a complete lattice.
\item
Each operator $\mathrm{Lift}_v$ is inflationary and monotone.
\item
From every $\mu : V \to S_{\eta, d}^\top$, 
every sequence of applications of operators $\mathrm{Lift}_v$
eventually reaches the least simultaneous fixed point of all
$\mathrm{Lift}_v$ that is greater than or equal to~$\mu$.
\item
A mapping $\mu : V \to S_{\eta, d}^\top$ 
is a simultaneous fixed point of all operators $\mathrm{Lift}_v$
if and only if it is a succinct progress measure.
\item
If $\mu^*$ is the least succinct progress measure,
then $\{v \,:\, \mu^*(v) \neq \top\}$ is the set of winning positions
for Even, and any choice of edges progressive in $\mu^*$, at least one
going out of each vertex she owns, is her winning positional strategy.  
\end{enumerate}
\end{theorem}

\begin{proof}
\begin{enumerate}
\item 
  The partial order of all mappings $V \to S_{\eta, d}^\top$ is 
  the pointwise product of $n$ copies of 
  the finite linear order $S_{\eta, d}^\top$.

\item
  We have inflation, i.e.\ $\mathrm{Lift}_v(\mu)(u) \geq \mu(u)$, by  
  the definitions of $\mathrm{Lift}_v(\mu)(u)$ and
  $\mathrm{lift}(\mu, v, w)$.
  
  For monotonicity, supposing $\mu \leq \mu'$, 
  it suffices to show that, for every edge $(v, w)$, we have
  $\mathrm{lift}(\mu, v, w) \leq \mathrm{lift}(\mu', v, w)$.
  Writing $\sigma'$ for $\mathrm{lift}(\mu', v, w)$,
  we know that $\sigma' \geq \mu'(v) \geq \mu(v)$.
  Also $(v, w)$ is progressive in $\mu'[v \mapsto \sigma']$, giving us
  that 
  \[\sigma'|_{\pi(v)} \geq
    \mu'[v \mapsto \sigma'](w)|_{\pi(v)} \geq
     \mu[v \mapsto \sigma'](w)|_{\pi(v)}\]
  and the first inequality is strict when $\pi(v)$ is odd
  unless $\mu'[v \mapsto \sigma'](w) = \top$;
  but if $\mu'[v \mapsto \sigma'](w) = \top$
      and $\mu[v \mapsto \sigma'](w) \neq \top$
  then the second inequality is strict,
  so in any case $(v, w)$ is progressive in $\mu[v \mapsto \sigma']$.
  Therefore $\mathrm{lift}(\mu, v, w) \leq \sigma'$.

\item
  This holds for any family of 
  inflationary monotone operators on a finite complete lattice.
  Consider any such maximal sequence from $\mu$.
  It is an upward chain from $\mu$ to some $\mu^*$ 
  which is a simultaneous fixed point of all the operators.
  For any $\mu' \geq \mu$ which is also a simultaneous fixed point,
  a simple induction confirms that $\mu^* \leq \mu'$.

\item
  Here we have a rewording of the definition of a succinct progress measure,
  cf.~Section~\ref{section:spm}.

\item
  The set of winning positions for Even is contained in
  $\{v \,:\, \mu^*(v) \neq \top\}$ by Lemma~\ref{lemma:necessity}
  because
  $\mu^*$ is the least succinct progress measure. 
  
  Since $\mu^*$ is a succinct progress measure, 
  we have that, for every progressive edge $(v, w)$, 
  if $\mu^*(v) \neq \top$ then $\mu^*(w) \neq \top$.
  It remains to apply Lemma~\ref{lemma:sufficiency} to 
  the subgame consisting of the vertices $\{v \,:\, \mu^*(v) \neq \top\}$, 
  the chosen edges from vertices owned by Even, 
  and all edges from vertices owned by Odd.
  \qedhere
\end{enumerate}
\end{proof}

\begin{table}
  \fbox{\parbox{0.98\columnwidth}{
  \begin{enumerate}
  \item
    Initialise $\mu : V \to S_{\eta, d}^\top$ so that it maps every
    vertex $v \in V$ to the bottom element in $S_{\eta, d}^\top$, which
    is the empty sequence. 

  \item
    While $\mathrm{Lift}_v(\mu) \neq \mu$ for some $v$,
    update $\mu$ to become $\mathrm{Lift}_v(\mu)$.

  \item
    Return the set 
    $W_{\mathrm{Even}} = \{v \: : \: \mu(v) \not= \top\}$
    of winning positions for Even, 
    and her positional winning strategy
    that for every vertex $v \in W_{\mathrm{Even}}$ owned by Even
    picks an edge outgoing from~$v$ that is progressive in~$\mu$.
  \end{enumerate}
  }}
  \caption{The lifting algorithm}
  \label{table:algorithm}
\end{table}

Note that the algorithm in Table~\ref{table:algorithm} is a solution to 
both variants of the algorithmic problem of solving parity games:
it finds the winning positions and produces a positional winning
strategy for Even.

\subsection{Algorithm analysis}

The following lemma offers various estimates for the size of the
set~$S_{\eta, d}$ of succinct adaptive multi-counters used in the lifting
algorithm in Table~\ref{table:algorithm}, and which is the dominating
factor in the worst-case upper bounds on the running time of the
algorithm.  
A particular focus is the analysis pinpointing the range of the
numbers of distinct priorities~$d$
(measured as functions of the number~$\eta$ of vertices with an odd
priority) 
in which the algorithm may cease to be polynomial-time.
The ``phase transition'' occurs when $d$ is logarithmic in~$\eta$: 
if $d = o(\log \eta)$ then the size of $S_{\eta, d}$ is
$O\left(\eta^{1+o(1)}\right)$, 
if $d = \Theta(\log \eta)$ then the size of $S_{\eta, d}$ is bounded
by a polynomial in~$\eta$ but its degree depends on 
the constant hidden in the big-$\Theta$, 
and if $d = \omega(\log \eta)$ then the size of~$S_{\eta, d}$ is 
superpolynomial in~$\eta$.  

\begin{lemma}[Size of $S_{\eta, d}$]
  \mbox{} 
  \label{lemma:size-of-Snd}
  \begin{enumerate}
  \item
    \label{enumerate:size-of-Bnd-binom}
    $|S_{\eta, d}| \leq 2^{\lceil \lg \eta \rceil}
    \binom{\lceil \lg \eta \rceil + d/2 + 1}{d/2}$.

  \item
    If $d = O(1)$ then $|S_{\eta, d}| = O\left(\eta \lg^{d/2} \eta\right)$.  

  \item
    \label{enumerate:d-delta-log-n}
    If $d/2 = \lceil \delta \lg \eta \rceil$,
    for some constant $\delta > 0$, then 
    $|S_{\eta, d}| =
    \Theta\left(\eta^{\lg(\delta+1) + \lg(e_\delta) + 1} \middle/
      \sqrt{\log \eta}\right)$, 
    where $e_{\delta} = (1 + 1/\delta)^\delta$. 

  \item
    If $d = o(\log \eta)$ then $|S_{\eta, d}| = O\left(\eta^{1+o(1)}\right)$. 

  \item
    If $d = O(\log \eta)$ then $|S_{\eta, d}|$ is bounded by a
    polynomial in~$\eta$.

  \item
    \label{enumerate:d-omega-log-n}
    If $d = \omega(\log \eta)$ then $|S_{\eta, d}|$ is superpolynomial
    in~$\eta$ and
    $|S_{\eta, d}| = O\left(d \eta^{\lg (d/{\lg \eta}) + 1.45}\right)$. 
  \end{enumerate}  
\end{lemma}

\begin{proof}
  \begin{enumerate}
  \item
    There are
    $2^{\lceil \lg \eta \rceil}$ bit sequences of
    length~$\lceil \lg \eta \rceil$ and for every
    $i, 0 \leq i \leq d/2$, there are 
    \[
    \binom{\lceil \lg \eta \rceil + i}{i}
    =
    \binom{\lceil \lg \eta \rceil + i}{\lceil \lg \eta \rceil}
    \]
    distinct ways of distributing the number $\lceil \lg \eta \rceil$
    of bits to $i+1$ components
    (the $i$ components in the succinct adaptive $i$-counter,
    and an extra one for the ``unused'' bits).
    By the ``parallel summation'' binomial identity, we obtain 
    \[
    \sum_{i=0}^{d/2} \binom{\lceil \lg \eta \rceil + i}{i} =
    \binom{\lceil \lg \eta \rceil + d/2 + 1}{d/2}\,.
    \]

    In cases~\ref{enumerate:d-delta-log-n})
    and~\ref{enumerate:d-omega-log-n}) below 
    we analyse the simpler expressions 
    $\binom{\lceil \lg \eta \rceil + d/2}{d/2}$ and
    $\binom{\lceil \lg \eta \rceil + d/2}{\lceil \lg \eta \rceil}$,
    respectively, instead of the above 
    $\binom{\lceil \lg \eta \rceil + d/2 + 1}{d/2}$, 
    in order to declutter calculations.
    This is justified because in each context the respective simpler 
    expression is within a constant factor of the latter one, and
    hence the asymptotic results are not affected by the
    simplification.   

  \item
    This is easy to verify for $d = 2$ and $d = 4$.
    If $d \geq 6 > 2e$ then for sufficiently large~$\eta$ we have:
    \[
      \binom{\lceil \lg \eta \rceil + d/2 + 1}{d/2}
      \leq
      \big((\lceil \lg \eta \rceil + d/2 + 1) \cdot (2e/d)\big)^{d/2}
      \leq
      \lceil \lg \eta \rceil^{d/2}\;.
    \]
    The former inequality always holds by the inequality 
    $\binom{k}{\ell} \leq \left(\frac{ek}{\ell}\right)^\ell$ applied to
    the binomial coefficient
    $\binom{\lceil \lg \eta \rceil + d/2 + 1}{d/2}$. 
    The latter inequality holds for sufficiently large~$\eta$
    because---by the assumption that $d > 2e$---we have that
    $2e/d < 1$, and hence the inequality holds for all~$\eta$
    large enough that $d/2 + 1 \leq (1-2e/d) \lceil \lg \eta \rceil$.  
  
  \item
    To avoid hassle, consider only the values of~$\eta$ and~$\delta$,
    such that both $\lg \eta$ and $\delta \lg \eta$ are integers. 
    Let $d = 2 \delta \lg \eta$ and
    apply~\cite[Lemma 4.7.1]{Ash90}
    (reproduced as Lemma~\ref{lemma:ash} in the Appendix) to the
    binomial coefficient  
    \[
    \binom{\lg \eta + d/2}{d/2} =
    \binom{\lg \eta + \delta \lg \eta}{\delta \lg \eta} =
    \binom{(\delta+1)\lg \eta}{\delta \lg \eta}\,,
    \]
    obtaining 
    \[
    |S_{\eta, d}| =
    \Theta\left(\eta^{(\delta+1) \,H\left(\tfrac{\delta}{\delta+1}\right) + 1}
                \middle/\sqrt{\log \eta}\right)\,, 
    \]
    where $H(p) = -p \lg p - (1-p) \lg (1-p)$ is the bi\-na\-ry
    entropy function, defined for $p \in [0, 1]$.
    A skilful combinator will be able to verify the identity
    \[
    (\delta+1) \:H\left(\frac{\delta}{\delta+1}\right)
    =
    \lg(\delta+1) + \lg(e_\delta)\,. 
    \]
    
  \item
    This is a corollary of part~\ref{enumerate:d-delta-log-n}) by
    observing that    
    $\lim_{\delta \downarrow 0} e_{\delta} = 1$ and hence: 
    \[
    \lim_{\delta \downarrow 0}
    \left(\lg(\delta+1) + \lg(e_{\delta}) + 1\right)
    = 1\,. 
    \]

  \item
    Again, this is a corollary of part~\ref{enumerate:d-delta-log-n})
    by observing that    
    the expression $\lg(\delta+1) + \lg(e_\delta) + 1$ is $O(1)$ 
    as a function of~$\eta$.

  \item
    The first statement is a corollary of
    part~\ref{enumerate:d-delta-log-n}) by observing that  
    $\lim_{\delta \to \infty} \lg(\delta+1) = \infty$ and 
    $\lim_{\delta \to \infty} e_{\delta} = e$, and hence: 
    \[
    \lim_{\delta \to \infty}
    \left(\lg(\delta+1) + \lg(e_{\delta}) + 1\right)
    = \infty\,. 
    \]

    In order to prove the latter statement, note that
    \begin{multline*}
      \lg \binom{\lceil \lg \eta \rceil + d/2}{\lceil \lg \eta \rceil} 
      \leq
      \lceil \lg \eta \rceil \cdot
      \Big[\lg\big(\lceil \lg \eta \rceil + d/2\big) -
        \lg \lceil \lg \eta \rceil + \lg e\Big]
      = 
      \\
      \lceil \lg \eta \rceil \cdot
      \Big[\lg\big((1+o(1)) d/2\big) -
        \lg \lceil \lg \eta \rceil + \lg e\Big]
      =
      \\
      \lceil \lg \eta \rceil \cdot
      \big[\lg d - \lg \lceil \lg \eta \rceil + \lg(e/2) + o(1)\big]\,,
    \end{multline*}
    where the first inequality is obtained by taking the $\lg$ of both  
    sides of the inequality 
    $\binom{k}{\ell} \leq \left(\frac{ek}{\ell}\right)^\ell$
    applied to the binomial coefficient
    $\binom{\lceil \lg \eta \rceil + d/2}{\lceil \lg \eta \rceil}$, and 
    the second relation follows from the assumption that
    $d = \omega(\log \eta)$.

    Then we have
    \begin{multline*}
    |S_{\eta, d}| =
    O\left(2^{\lceil \lg \eta \rceil \cdot \big(1 + \lg d - \lg \lg \eta + \lg(e/2) + o(1)\big)}\right) =
    \\
    O\left(2^{(1 + \lg \eta) \cdot \big(\lg d - \lg \lg \eta + \lg e + o(1)\big)}\right) =
    O\left(d \eta^{\lg(d/{\lg \eta}) + 1.45}\right)
    \,,
    \end{multline*}
    where the latter holds because
    \[
    2^{\lg d - \lg \lg \eta + O(1)} = O(d/{\lg \eta})\,,
    \]
    and $\lg e + o(1) < 1.4427$ for sufficiently large~$\eta$. 
    \qedhere
  \end{enumerate}
\end{proof}

\begin{theorem}[Complexity of lifting algorithm]
\label{theorem:complexity}
  \mbox{} 
  \begin{enumerate}
  \item
    If $d = O(1)$ then the algorithm runs in time
    $O\left(m \eta \lg^{d/2+1} \eta\right)$.

  \item
    If $d = o(\log \eta)$ then the algorithm runs in time
    $O(m \eta^{1+o(1)})$.

  \item
    \label{enumerate:d-less-lg-n}
    If $d \leq 2 \lceil \delta \lg \eta \rceil$, for some positive
    constant~$\delta$, then the algorithm runs in time
    \[
    O\left(m \eta^{\lg(\delta+1) + \lg(e_\delta) + 1}
           \sqrt{\log \eta} \log \log \eta\right)\,.
    \]
    In particular, if $d \leq \lceil \lg \eta \rceil$, then the running 
    time is
    $O\left(m \eta^{2.38}\right)$.

  \item
    \label{enumerate:d-omega-lg-n}
    If $d = \omega(\log \eta)$ then the algorithm runs in time
    $O\left(dm \eta^{\lg(d/{\lg \eta}) + 1.45}\right)$.
  \end{enumerate}
  The algorithm works in space $O(n \log n \cdot \log d)$.
\end{theorem}

\begin{proof}
The work space requirement is dominated by the number of bits needed 
to store a single mapping $\mu : V \to S_{\eta, d}^\top$, which is at
most $n \lceil \lg \eta \rceil \lceil \lg d \rceil$. 

We claim that the $\mathrm{Lift}_v$ operators can be implemented to work in time
$O(\mathrm{deg}(v) \cdot \log \eta \cdot \log d)$.
It then follows, since the algorithm lifts each vertex at most 
$|S_{\eta, d}|$ times, that its running time is bounded by
\[
  O\left(\sum_{v \in V} 
    \mathrm{deg}(v) \cdot \log \eta \cdot \log d \cdot
    |S_{\eta, d}|\right) =
  O\left(m \log \eta \cdot \log d \cdot |S_{\eta, d}|\right)\,.
\]
From there, the various stated bounds are obtained by Lemma~\ref{lemma:size-of-Snd}.
For the last statement in part~\ref{enumerate:d-less-lg-n}), note that
if $\delta = 1/2$ then 
$d \leq \lceil \lg \eta \rceil$ implies
$d/2 \leq \lceil \delta \lg \eta \rceil$, and
\[
\lg(\delta+1) + \lg(e_\delta) + 1 = \frac{3}{2} \lg 3
< 2.3775\,.
\]

To establish the claim, it suffices to observe that every 
bounded adaptive multi-counter 
$\mathrm{lift}(\mu, v, w) \in S_{\eta, d}^\top$ 
is computable in time $O(\log \eta \cdot \log d)$.
The computation is most involved when $\pi(v)$ is odd and
$\mu(v)|_{\pi(v)} \leq \mu(w)|_{\pi(v)} \neq \top$,
which imply that $\mathrm{lift}(\mu, v, w)$ is the least 
$\sigma \in S_{\eta, d}^\top$ such that
$\sigma|_{\pi(v)} > \mu(w)|_{\pi(v)}$.
Writing $(s_{d-1}, s_{d-3}, \ldots, s_{k+2}, s_k)$ for $\mu(w)$, 
there are five cases:
\begin{itemize}
\item
  If $k > \pi(v)$, 
  then obtain $\sigma$ as  
  \[
  (s_{d-1}, \ldots, s_{k+2}, s_k, 0 \cdots 0)\,,
  \]
  where the padding by $0$s is up to the total length
  $\lceil \lg \eta \rceil$ of~$\sigma$.

\item
  If $k \leq \pi(v)$ and the total length of $s_i$ for $i \geq \pi(v)$
  is less than $\lceil \lg \eta \rceil$, then obtain $\sigma$ as  
  \[
  (s_{d-1}, \ldots, s_{\pi(v)+2}, s_{\pi(v)}10 \cdots 0)\,,
  \]
  where the padding by $0$s (if any) is up to the total length
  $\lceil \lg \eta \rceil$ of~$\sigma$.

\item
  If the total length of $s_i$ for $i \geq \pi(v)$ equals
  $\lceil \lg \eta \rceil$,
  $j$~is the least odd priority such that $s_j \neq \varepsilon$  
  (in this case, necessarily, $j \geq \pi(v)$),
  and $s_j$ is of the form $s' 0 \overbrace{1 \cdots 1}^\ell$ 
  (where possibly $\ell = 0$), then obtain $\sigma$ as
  \[
  \left(s_{d-1}, \ldots, s_{j+4}, s_{j+2}, s'\right)\,.
  \]
  
\item
  If the total length of $s_i$ for $i \geq \pi(v)$ equals
  $\lceil \lg \eta \rceil$,
  $j$ is the least odd priority such that $s_j \neq \varepsilon$
  (again, $j \geq \pi(v)$),
  $s_j$ is of the form $\overbrace{1 \cdots 1}^\ell$,
  and $j < d - 1$, then obtain $\sigma$ as
  \[
  \left(s_{d-1}, \ldots, s_{j + 4},
    s_{j + 2} 1 \overbrace{0 \cdots 0}^{\ell -1}\right)\,.
  \] 

  \item
    Otherwise, $\sigma = \top$.
    \qedhere
\end{itemize}
\end{proof}

\begin{corollary}[\cite{CJKLS17}]
Solving parity games is in FPT.
\end{corollary}

\begin{proof}
The algorithm runs in time
$\max\left\{
  2^{O(d \log d)},
  O\left(m n^{2.38}\right)
  \right\}$.
\end{proof}

\section{Separating automata}

Boja\'nczyk and Czerwi\'nski~\cite[Chapter~3]{BC18} have recently developed
a reworking of the algorithm of Calude et al.~\cite{CJKLS17} that proceeds by:
\begin{itemize}
\item 
constructing a deterministic safety automaton that \emph{separates}
the language of all infinite words of vertices in which all cycles are even 
from its odd counterpart;
\item 
forming a safety game as a synchronous product of 
the given parity game and the constructed separating automaton, 
in which the winning condition for Even is the acceptance condition of the automaton;
\item 
solving the formed safety game.
\end{itemize}
Their approach clarifies that the bulk of the quasi-polynomial time breakthrough
can be seen as showing how to construct a separating automaton of quasi-polynomial size.

In the rest of this section, we develop separating automata 
of quasi-polynomial size that are based on the bounded adaptive multi-counters.

Working with notations $V$, $n$, $\pi$, $d$ and $\eta$ as before, let us say that:
\begin{itemize}
\item 
a \emph{cycle} in a word of vertices is an infix whose first and last elements are the same;
\item 
the language $\alleven_{V, \pi}$ consists of all infinite words of vertices
in which all cycles are even;
\item 
the language $\limsupodd_{V, \pi}$ consists of all infinite words of vertices
in which the highest priority occurring infinitely often is odd;
\item 
the set $S_{\eta, d}^\bot$ is the linearly ordered set $S_{\eta, d}$
of bounded adaptive multi-counters, with an extra bottom element $\bot$
whose truncations are defined as $\bot|_p = \bot$;
\item 
a pair $(\sigma, \tau)$ of multi-counters is \emph{progressive} 
with respect to a priority $p$ if and only if: $\sigma|_p \geq \tau|_p$ and, 
when $\pi(v)$ is odd, either the inequality is strict or $\sigma = \tau = \bot$.
\end{itemize}

Let $\mathcal{D}_{V, \pi}$ be the following deterministic automaton 
that reads infinite words of vertices:
\begin{itemize}
\item 
the set of states is $S_{\eta, d}^\bot$, the initial state is the maximum multi-counter
\[\left(\overbrace{1 \cdots 1}^{\lceil \lg \eta \rceil}, 
        \overbrace{\varepsilon, \ldots, \varepsilon}^{d/2 - 1}\right)\;;\]
\item 
from state $\sigma$, reading $v$ leads to the greatest state $\tau$ such that
$(\sigma, \tau)$ is progressive with respect to~$\pi(v)$;
\item 
the only unsafe state is~$\bot$.
\end{itemize}
The automaton accepts if and only if its run is safe, i.e.\ does not visit an unsafe state.

The property we now establish implies that the automaton separates 
the language of all infinite words of vertices in which all cycles are even 
from its odd counterpart, because the latter is included in 
the language where the highest priority occuring infinitely often is odd.

\begin{theorem}
The automaton $\mathcal{D}_{V, \pi}$ 
accepts all words in the language $\alleven_{V, \pi}$ and 
rejects all words in the language $\limsupodd_{V, \pi}$.
\end{theorem}

\begin{proof}
That every word with the highest priority occuring infinitely often odd 
is rejected can be seen straightforwardly, 
since more than $2^{\lceil \lg \eta \rceil}$ occurences of such a priority in a word 
without intermediate occurences of higher priorities 
necessarily cause the multi-counters to underflow to~$\bot$.

The interesting half of the statement follows from the next claim 
by monotonicity of the truncation operations.  
Here $\mathcal{N}_{V, \pi}$ is the nondeterministic extension 
of the automaton $\mathcal{D}_{V, \pi}$ by replacing 
the `greatest' requirement for the successor states with `any'.  
Being nondeterministic, $\mathcal{N}_{V, \pi}$ accepts an infinite word 
if and only if some run on it is safe.

\begin{claim}
On every finite or infinite word over $V$ in which all cycles are even,
the automaton $\mathcal{N}_{V, \pi}$ has a safe run.
\end{claim}

\begin{proof}
We establish the claim by an induction on $\lceil \lg \eta \rceil$ and $d/2$ that follows 
the same pattern as the proof of Lemma~\ref{lemma:succinct-tree-coding}.

The base case, $\eta = 1$ and $d/2 = 0$, is trivial.

It suffices to consider words $\varpi$ that 
contain no vertices of the highest even priority~$d$.
Since $\varpi$ has no odd cycles, it can be decomposed as 
$\varpi_1 \, \varpi_\varepsilon \, \varpi_0$, where:
\begin{itemize}
\item 
if the word $\varpi_1$ is nonempty, 
then it ends with a vertex of the highest odd priority $d - 1$
and the set $V_1$ of all other odd-priority vertices that occur in $\varpi_1$ 
has cardinality at most~$\eta/2$;
\item 
the set $V_\varepsilon$ of all vertices 
that occur in the word $\varpi_\varepsilon$
contains no vertex of priority~$d - 1$;
\item 
if the word $\varpi_0$ is nonempty, 
then it begins with a vertex of the highest odd priority $d - 1$
and the set $V_0$ of all other odd-priority vertices that occur in $\varpi_0$ 
has cardinality at most~$\eta/2$.
\end{itemize}

It remains to compose a safe run 
of the automaton $\mathcal{N}_{V, \pi}$ 
on the word $\varpi$ by concatenating the following subruns, 
where we focus on the most involved case of $\varpi_1$ and $\varpi_0$ both nonempty: 
\begin{itemize}
\item 
obtain a safe run of the automaton $\mathcal{N}_{V_1, \pi}$ 
on the word $\varpi_1$ without its last vertex by the inductive hypothesis, 
and append one leading $1$ to the first binary strings in all its states;
\item 
obtain a safe run of the automaton $\mathcal{N}_{V_\varepsilon, \pi}$ 
on the word $\varpi_\varepsilon$ by the inductive hypothesis, 
and insert $\varepsilon$ as the first binary string in all its states;
\item 
obtain a safe run of the automaton $\mathcal{N}_{V_0, \pi}$ 
on the word $\varpi_0$ without its first vertex by the inductive hypothesis, 
and append one leading $0$ to the first binary strings in all its states.
\qedhere
\end{itemize}
\end{proof}

That also completes the proof the theorem.
\end{proof}

\section*{Acknowledgements}

We thank Kousha Etessami, John Fearnley, Filip Mazowiecki, and Sven Schewe 
for helpful comments; and Adam Lewis (a second-year undergraduate at the time) 
for finding and fixing a bug in the proof of Theorem~\ref{theorem:complexity}.

\bibliographystyle{plain}
\bibliography{smaller}


\section*{Appendix}

\subsection*{Estimates for binomial coefficients}

We outsource the challenge---and the tedium---of rigorously applying 
Stirling's approximation to estimating binomial coefficients
$\binom{k}{\ell}$, where $\ell = \Theta(k)$, to Ash~\cite{Ash90}. 
The following is Lemma~4.7.1 from page~113 in his book.

\begin{lemma}[Estimating binomial coefficients~\cite{Ash90}]
  \label{lemma:ash}
  If $0 < p < 1$
  and $pk$ is an integer, then 
  \[
  \frac{2^{k H(p)}}{\sqrt{8p(1-p)k}}
  \leq
  \binom{k}{pk}
  \leq
  \frac{2^{k H(p)}}{\sqrt{2\pi p(1-p)k}}\;,
  \]
  where $H(p) = -p \lg p - (1-p) \lg (1-p)$ is the binary
  entropy function.   
\end{lemma}

\end{document}